\documentclass[runningheads]{llncs}

\usepackage{amsmath, amssymb} 
\usepackage{textgreek}
\usepackage{color}
\usepackage{caption}
\usepackage{subcaption}
\usepackage[ruled,vlined,linesnumbered]{algorithm2e}
\usepackage{enumitem}
\usepackage{multirow}
\usepackage{float}

\DeclareMathOperator{\obs}{obs}

\usepackage{graphicx}
\begin{document}
\title{Interpretable exact linear reductions\\ via positivity\thanks{Supported by the Paris Ile-de-France region. GP was partially supported by NSF grants DMS-1853482, DMS-1760448, DMS-1853650, CCF-1564132, and CCF-1563942.}}
%
%
\author{Gleb Pogudin\inst{1} \and
Xingjian Zhang\inst{2}}
\authorrunning{G. Pogudin and X. Zhang}
\institute{LIX, CNRS, \'Ecole Polytechnique, Institute Polytechnique de Paris, Palaiseau, France \email{gleb.pogudin@polytechnique.edu} \and
\'Ecole Polytechnique, Institute Polytechnique de Paris, Palaiseau, France \email{xingjian.zhang@polytechnique.edu}}
\maketitle              
\begin{abstract}

Kinetic models of biochemical systems used in the modern literature often contain hundreds or even thousands of variables.
While these models are convenient for detailed simulations, their size is often an obstacle to deriving mechanistic insights.
One way to address this issue is to perform an exact model reduction by finding a self-consistent lower-dimensional projection of the corresponding dynamical system.

Recently, a new algorithm CLUE~\cite{Ovchinnikov2021} has been designed and implemented, which allows one to construct an exact linear reduction of the smallest possible dimension such that the fixed variables of interest are preserved.
It turned out that allowing arbitrary linear combinations (as opposed to zero-one combinations used in the prior approaches) may yield a much smaller reduction.
However, there was a drawback: some of the new variables did not have clear physical meaning, thus making the reduced model harder to interpret.

We design and implement an algorithm that, given an exact linear reduction, re-parametrizes it by performing an invertible transformation of the new coordinates to improve the interpretability of the new variables.
We apply our algorithm to three case studies and show that ``uninterpretable'' variables disappear entirely in all the case studies.
	    
The implementation of the algorithm and the files for the case studies are available at \url{https://github.com/xjzhaang/LumpingPostiviser}.

\keywords{exact reduction (lumping) \and ODE model \and interpretability.}

\end{abstract}
\section{Introduction}

Dynamical models described by systems of polynomial ordinary differential equations (PODEs) are frequently used in systems biology and life sciences in general.
One of the major classes of such models is the dynamical models of chemical reaction networks (CRN) under the mass-action kinetics in which each indeterminate corresponds to the concentration of one of the chemical species.
Models appearing in the literature often consist of hundreds or thousands of variables.
While the models of this size can incorporate a substantial amount of information about the phenomena of interest, it is often hard to use them to derive mechanistic insights.

One way to address these challenges is to use \emph{model reduction} algorithms that replace a model with a simpler one while preserving, at least approximately, some of the features of the original model.
A wide range of methods has been developed for approximate model reduction, including methods based on singular value decomposition~\cite{antoulas} and time-scale separation~\cite{okino1998}.

A complementary approach is to perform~\emph{exact model reduction}, that is, lower the dimension of the model without introducing approximation errors.
For example, exact linear lumping aims at writing a self-consistent system of differential equations for a set of \emph{macro-variables} in which each macro-variable is a linear combination of the original variables.
For important classes of biochemical models, specialized lumping criteria have been developed (see, e.g., ~\cite{Borisov,Conzelmann,Feret}), allowing the construction of macro-variables as sums of some of the original variables (that is, allowing only coefficients zero and one in the linear combinations).
A general lumping algorithm has been proposed in~\cite{Cardelli2017,erode} which is applicable to any system of PODEs (not necessarily arising from a CRN).
This algorithm partitions the original variables so that the macro-variables can be the sums of the variables within the blocks in the partition.
Note that the macro-variables are zero-one linear combinations of the original variables in all these cases.

In~\cite{Ovchinnikov2021}, an algorithm has been designed (and the corresponding software called \emph{CLUE} presented) that, for a given set of linear forms in the state variables (the \emph{observables}), constructs a linear lumping of the smallest possible dimension such that the observables can be written as combinations of the macro-variables (i.e., the observables are preserved).
Unlike the earlier approaches, the macro-variables produced by CLUE may involve any coefficients, and this allowed to produce reductions of lower dimensions than it was possible before, see~\cite[Table~1]{Ovchinnikov2021}.
However, there was a price to pay for this flexibility: the authors state that some of the produced macro-variables ``escape physical intelligibility'' (see~\cite[Section~4.2]{Ovchinnikov2021}).
Indeed, the resulting reduction of the smallest dimension is uniquely defined up to a linear change of the coordinates, so the coordinates in the reduced state space chosen by CLUE could be not optimal in the sense of interpretability.

In this paper, we propose a post-processing step that takes an exact linear lumping (not necessarily produced by CLUE) and attempts to improve its interpretability by performing a change of variables.
It has been observed in~\cite{Ovchinnikov2021} that one of the sources of difficulties for interpretation is the negative coefficients in the macro-variables.
We design and implement an algorithm that finds (if possible) a linear change of variables in the reduced model so that
\begin{enumerate}
  \item the coefficients of the representations of the new macro-variables in terms of the original state variables are nonnegative

  \item and the total number of nonzero coefficients in these representations is as small as possible.
\end{enumerate}
Note that interpretability is not a formal mathematical property, and the conditions above is one possible formalization of the notion of a ``more interpretable reduction''.
We do not claim that it is universal (e.g., a difference of two state variables may represent a potential), but we claim that it is useful.
To support this claim, we demonstrate the efficiency of our approach on three case studies from the literature.
Two of these cases are exactly the case studies from~\cite{Ovchinnikov2021} in which issues with interpretability occur.
We show that our method provides interpretable re-parametrizations of the optimal lumpings computed by CLUE in all three case studies.
Our algorithm uses tools from convex discrete geometry and matroid theory.

%
%
%
\section{Methods}

\subsection{Preliminaries on lumping}

\begin{definition}[Lumping]\label{def:lumping}
Consider a system of ODEs of the form
  \begin{equation}\label{eq:main}
    \mathbf{x}' = \mathbf{f}(\mathbf{x}),
  \end{equation}
  where $\mathbf{x} = (x_1, \ldots, x_n)^T$, $\mathbf{f} = (f_1, \ldots, f_n)^T$, and $f_1, \ldots, f_n \in \mathbb{R}[\mathbf{x}]$.
  A linear transformation $\mathbf{y} = L\mathbf{x}$ with $\mathbf{y} = (y_1, \ldots, y_m)^T$, $L \in \mathbb{R}^{m \times n}$, and $\operatorname{rank}L = m$ is called \emph{a lumping of~\eqref{eq:main}} if there exist polynomials $g_1, \ldots, g_m \in \mathbb{R}[\mathbf{y}]$ such that 
  \[
    \mathbf{y}' = \mathbf{g}(\mathbf{y}), \quad\text{ where }\quad \mathbf{g} = (g_1, \ldots, g_m)^T
  \]
  for every solution $\mathbf{x}$ of~\eqref{eq:main}. We say that $m$ is \emph{the dimension of the lumping}. The variables $\mathbf{y}$ in the reduced system are called \emph{macro-variables}. 
  We will call a macro-variable \emph{nontrivial} if it is not proportional to one of the original variables.
\end{definition}

\begin{remark}
  An ODE system may have many lumpings, some of them may be less useful than others.
  For example, if $m = n$, then the lumping is just an invertible change of variables, so no reduction happens.
  Another special case is when the rows of $L$ contain the coefficients of linear first integrals of the system.
  In this case, the reduced ODE will be of the form $\mathbf{y}' = 0$.
  
  \emph{Constrained linear lumping} introduced in Definition~\ref{def:cll} requires to preserve the dynamics of the variables of interest, and this is one of the ways to say that reduction is not ``too coarse''.
\end{remark}

The following example is a substantially simplified version of the case study from Section~\ref{subsec:case1} (see also~\cite{Gunawardena2005}).

\begin{example}\label{ex:lumping}
  We will consider a chemical reaction network consisting of 
  \begin{itemize}
      \item A chemical species $X$.
      \item Species $A_{UU}$, $A_{UX}$, $A_{XU}$, and $A_{XX}$.
       Each of them is one of the states of a molecule $A$ with two identical binding sites, which can be either unbound (U in the subscript) or bound (X in the subscript) to $X$.
  \end{itemize}
  For simplicity, we will assume that all the reaction rates are equal to one.
  The dynamics of the network is defined by the following reactions ($\ast$ denotes any of $X$ and $U$):
  \begin{equation}\label{eq:ex_reactions}
    X + A_{U\ast} \rightleftharpoons A_{X\ast}, \qquad X + A_{\ast U} \rightleftharpoons A_{\ast X}.
  \end{equation}
  Under the laws of the mass-action kinetics, the reactions~\eqref{eq:ex_reactions} yield the following ODE system (where $[S]$ denotes the concentration of the species $S$):
  \begin{equation}\label{eq:ex_ode}
      \begin{cases}
        [X]' = [A_{XU}] + [A_{UX}] + 2[A_{XX}] - [X] ([A_{XU}] + [A_{UX}] + 2[A_{UU}]),\\
        [A_{UU}]' = [A_{XU}] + [A_{UX}] - 2[X][A_{UU}],\\
        [A_{XU}]' = [A_{XX}] + [X][A_{UU}] - [X][A_{XU}] - [A_{XU}],\\
        [A_{UX}]' = [A_{XX}] + [X][A_{UU}] - [X][A_{UX}] - [A_{UX}],\\
        [A_{XX}]' = [X][A_{XU}] + [X][A_{UX}] - 2[A_{XX}].
      \end{cases}
  \end{equation}
  We will show that the following matrix $L$ and the macro-variables $y_1, y_2, y_3$
  \begin{equation}\label{eq:ex_lumping}
    L = \begin{pmatrix}
      1 & 0 & 0 & 0 & 0\\
      0 & 0 & 1 & 1 & 2\\
      0 & 2 & 1 & 1 & 0
    \end{pmatrix}
    \implies
    \begin{cases}
       y_1 = [X],\\
       y_2 = [A_{XU}] + [A_{UX}] + 2[A_{XX}],\\
       y_3 = 2[A_{UU}] + [A_{XU}] + [A_{UX}].
    \end{cases}
  \end{equation}
  yield a lumping of the system~\eqref{eq:ex_reactions}.
  Indeed, a direct calculation shows that
  \begin{equation}\label{eq:ex_result}
  \begin{cases}
     y_1' = [X]' = [A_{XU}] + [A_{UX}] + 2[A_{XX}] - [X] ([A_{XU}] + [A_{UX}] + 2[A_{UU}]) = y_2 - y_1y_3,\\
     y_2' = [A_{XU}]' + [A_{UX}]' + 2[A_{XX}]' = -y_2 + y_1y_3,\\
     y_3' = 2[A_{UU}]' + [A_{XU}]' + [A_{UX}]' = y_3 - y_1y_2.
  \end{cases}
  \end{equation}
  Since each reaction involves only one binding site, this lumping can be interpreted as follows: $y_2$ is the total ``concentration'' of the bound sites, and $y_3$ is the total ``concentration'' of the unbound sites (see also Section~\ref{subsec:case1}).
\end{example}

The lumping matrix $L$ in the example above turns out to exactly preserve the concentration $[X]$.
In general, one may fix a vector $\mathbf{x}_{\obs}$ of combinations of the original variables that are to be recovered in the reduced system.

\begin{definition}[Constrained linear lumping]\label{def:cll}
Let $\mathbf{x}_{\obs}$ be a vector of linearly independent forms in $\mathbf{x}$ such that  $\mathbf{x}_{\obs} = A\mathbf{x}$. 
Then we say that a lumping $\mathbf{y} = L\mathbf{x}$ is a \emph{constrained linear lumping} with observables $\mathbf{x}_{\obs}$
if each entry of $\mathbf{x}_{\obs}$ is a linear combination of the entries of $\mathbf{y}$. 
\end{definition}

\subsection{The nonuniqueness/interpretability issue}

A recent software CLUE~\cite{Ovchinnikov2021} allows to find, for a given system~\eqref{eq:main} and a vector $\mathbf{x}_{\obs}$, a constrained linear lumping of the smallest possible dimension.
However, such an optimal lumping is not unique in the following sense: if $\mathbf{y}_1 = L\mathbf{x}$ is a constrained linear lumping of the smallest possible dimension, then, for every invertible matrix $T$ of the appropriate dimension, $\mathbf{y}_2 = TL\mathbf{x}$ is also such a lumping. 
Two such lumpings will be called \emph{equivalent}, and one can show that all constrained linear lumpings of the smallest possible dimension are equivalent.

Because of this nonuniqueness, the lumping produced by CLUE will be optimal in terms of the dimension but not necessarily optimal in terms of the \emph{interpretability} of the resulting macro-variables.
For example, the macro-variables constructed by CLUE for the system~\eqref{eq:ex_ode} are:
\[
y_1 = [X],\quad y_2 = [A_{XU}] + [A_{UX}] + 2[A_{XX}],\quad y_3 = [A_{UU}] - [A_{XX}].
\]
The last macro-variable is different from the one in~\eqref{eq:ex_lumping} and does not allow for the ``concentration-of-sites'' interpretation.
Moreover, the reduced ODE system is more complicated than~\eqref{eq:ex_result}.
This issue becomes more serious in more realistic (and larger) models: for the case studies in~\cite[Section~4.2]{Ovchinnikov2021} it has been observed that some of the resulting macro-variables ``escaped physical intelligibility''.


\subsection{Our approach via nonnegativity}
\label{sec:our_approach}

It has been already observed in~\cite[Section~4.2]{Ovchinnikov2021} that the macro-variables involving negative coefficients (such as $[A_{UU}] - [A_{XX}]$) may be an obstacle for interpretability.
This is partially because such quantities cannot be naturally viewed as concentrations of some sort since they may take on negative values.

Thus, in order to improve the interpretability of a lumping, we construct an equivalent lumping with all the coefficients being nonnegative and the number of nonzero coefficients (that is, the $\ell_0$-norm $\lVert \cdot \rVert_0$) being the smallest possible under the nonnegativity constraint.
Mathematically, for a given lumping $\mathbf{y}_1 = L\mathbf{x}$, we find (if possible) an equivalent lumping $\mathbf{y}_2 = TL\mathbf{x}$ with invertible $T$ satisfying:
\begin{enumerate}
    \item the entries of $TL$ are nonnegative and
    \item $\lVert TL\rVert_0$ is as small as possible.
\end{enumerate}
As we have mentioned, for fixed observables, all the constrained linear lumpings of the smallest dimension are equivalent, so the value $\lVert TL \rVert_0$ does not depend on the choice of $L$ in the case of the optimal constrained linear lumping as in~\cite{Ovchinnikov2021}.

We hypothesize that the new lumping $\mathbf{y}_2 = TL\mathbf{x}$ will be typically more interpretable than the original one.
We support this hypothesis by three case studies: multisite protein phosphorylation~\cite{Sneddon2010}, Fc$\epsilon$-RI signaling pathways~\cite{Faeder2003}, and Jak-family protein tyrosine kinase activation~\cite{Barua2009}.
The first two are exactly the case studies from~\cite{Ovchinnikov2021} for which some of the macro-variables could not be properly interpreted by the authors.

\subsection{Algorithmic details}

In this section, we provide and justify Algorithm~\ref{alg:main}, an algorithm for computing a new lumping described in Section~\ref{sec:our_approach}.
We will use some basic terminology from convex geometry. We refer the reader to~\cite[Chapters 7-8]{convex} for details.
Throughout the rest of the section, for $A$ being a vector or a matrix, $\lVert A \rVert_0$ denotes the $\ell_0$-norm of $A$, that is, the number of nonzero entries in $A$.

{\small
\begin{algorithm}[H]
\caption{Algorithm for constructing new lumping}\label{alg:main}

\begin{description}
  \item[Input] a $m \times n$ matrix $L$ with entries in $\mathbb{R}$ and linearly independent rows;
  
  \item[Output] an invertible $m\times m$ matrix $T$ such that 
  \begin{itemize}
      \item the entries of $TL$ are nonnegative
      \item and the number of the nonzero entries is as small as possible.
  \end{itemize}
  Returns NO if such matrix $T$ does not exist.
\end{description}

\begin{enumerate}[label = \textbf{(Step~\arabic*)}, leftmargin=*, align=left, labelsep=2pt, itemsep=0pt]
    \item Consider the row space of $L$ and the nonegative orthant in $(\mathbb{R}_{\geqslant 0})^n$ as polyhedral cones $C_1$ and $C_2$ in $\mathbb{R}^n$.
    \item\label{step:intersection} Compute a polyhedral cone $C = C_1 \cap C_2$.
    This can be done, for example, using the Fourier-Motzkin algorithm~\cite[Section~1.2]{Ziegler}.
    \item If $\dim C < m$, return NO
    \item\label{step:producing} Let $E$ be a set of representatives of the extreme rays of $C$.
    \item Initialize a $0\times n$ matrix $L_1$
    \item\label{step:greedy} While $E \neq \varnothing$
    \begin{enumerate}
        \item\label{step:choice} choose $e \in E$ such that $\lVert e\rVert_0 = \min_{v \in E} \lVert v\rVert_0$;
        \item if $e$ is not in the row space of $L_1$, append $e$ to $L_1$ as a new row;
        \item remove $e$ from $E$.
    \end{enumerate}
    \item\label{step:conctruction} Construct an $m \times m$ matrix $T$ such that the $i$-th column contains the coordinates of the $i$-th row of $L_1$ with respect to the rows of $L$.
\end{enumerate}
\end{algorithm}
}
\begin{remark}[Implementation]
   Our implementation of Algorithm~\ref{alg:main} in Julia can be found at~\url{https://github.com/xjzhaang/LumpingPostiviser}.
   We used polymake~\cite{Gawrilow2000} for operations with cones (at \ref{step:intersection} and \ref{step:producing}) and Nemo~\cite{nemo} for symbolic linear algebra (at~\ref{step:greedy}).
   Table~\ref{table:runtimes} below summarizes the performance of the code on the case studies we discuss in this paper.
   We also provide timing for obtaining the starting reduction using CLUE. Therefore, the sum of the last two columns is the total time to obtain the final reduction for the original system.
   The runtimes are measured on a laptop with a 2.20GHz CPU and 16GB RAM using \texttt{@btime} macro in Julia.
   One can see that the models with hundreds of equations can be tackled in less than a minute on a commodity hardware.
   {\small
   \begin{table}[!htbp]
       \centering
       {\renewcommand{\arraystretch}{1.2}
       \begin{tabular}{|c|c|c|c|c|}
       \hline
            \multirow{2}{*}{Model} & \multirow{2}{*}{\# original variables ($n$)} & \multirow{2}{*}{\# macro-variables ($m$)} & \multicolumn{2}{c|}{Runtime (sec.)} \\
            \cline{4-5}
            & & & CLUE & Algorithm~\ref{alg:main}\\
            \hline
            Section~\ref{subsec:case1}, $m = 2$ & $18$ & $6$ & $< 0.01$ & $< 0.01$\\ 
            \hline
            Section~\ref{subsec:case1}, $m = 3$ & $66$ & $6$ &$ < 0.01$  &$< 0.01$\\
            \hline
            Section~\ref{subsec:case1}, $m = 4$ & $258$ & $6$ & $0.34$ & $3.4$\\ 
            \hline
            Section~\ref{subsec:fceri} & $354$ & $69$ & $3.3 $&$4.7$\\
            \hline
            Section~\ref{subsec:barua} & $470$ & $322$ & $72$ & $49$\\
            \hline
       \end{tabular}}
       \caption{Running times of our implementation. }
       \label{table:runtimes}
   \end{table}}
\end{remark}

\vspace{-12mm}
\begin{remark}[Choice at~\ref{step:choice}]
   At the~\ref{step:choice}, if there are several $e \in E$ with $\lVert e \rVert_{0}$ being minimal possible, we choose the one with the index of the leftmost nonzero entry being the smallest one.
   In our experience, this makes the results slightly easier to analyze.
\end{remark}

\begin{remark}[Returning NO]\label{rem:NO}
  Although Algorithm~\ref{alg:main} may, in principle, return NO, we did not encounter such a situation with models from the literature. 
  We give an artificial example with this property in Appendix.
\end{remark}

\begin{theorem}[Correctness of Algorithm~\ref{alg:main}]
  For every matrix $L$ over $\mathbb{R}$ with linearly independent rows, Algorithm~\ref{alg:main} produces an invertible square matrix $T$ such that
  \begin{itemize}[topsep=1pt]
      \item $TL$ has nonnegative entries
      \item and the number of nonzero entries in $TL$ is the smallest possible under the nonnegativity constraint
  \end{itemize}
  if such $T$ exists and returns NO if there is no such $T$.
\end{theorem}

\begin{proof}
  First, we will show that the algorithm returns NO if and only if there is no such matrix.
  Assume that there is such a matrix $T$.
  Then both $C_1$ and $C_2$ contain the rows of the matrix $TL$.
  Therefore, $C$ contains $m$ linearly independent vectors, so its dimension is at least $m$.
  In the other direction, if $\dim C \geqslant m$, then there exist $m$ linearly independent vectors in $C = C_1 \cap C_2$.
  Let $T$ be the matrix with the columns being their coordinates with respect to the rows of $L$.
  Then the rows of $TL$ will belong to $C_2$ so that they will be nonnegative.
  
  Now assume that the algorithm does not return NO.
  We observe that the entries of $L_1$ are nonnegative because all its rows belong to $C_2$.
  The rows of $L_1$ belong to $C_1$, so they are linear combinations of the rows of $L$.
  Since, by the construction on~\ref{step:greedy}, the rows of $L_1$ are linearly independent, and there are $\dim C$ of them, we conclude that the row spaces of $L_1$  and $L$ coincide.
  Therefore, the coordinates in~\ref{step:conctruction} are well-defined, so the algorithm will produce a matrix $T$ such that $L_1 = TL$ has only nonnegative entries.
  
  It remains to prove that the $\ell_0$-norm $\lVert L_1 \rVert_0$ of $L_1 = TL$ is the smallest possible.
  Consider any set $S$ of $m$ linearly independent elements of the set $E$ of representatives of the extreme rays of $C$.
  Since~\ref{step:conctruction} is a greedy algorithm on the linear matroid defined by $E$, \cite[(18)]{Edmonds1971} implies that 
  \begin{equation}\label{eq:l0norm}
      \lVert L_1 \rVert_0 \leqslant \sum\limits_{e \in S} \lVert e \rVert_0.
  \end{equation}
  Consider any invertible matrix $\widetilde{T}$ such that the entries of $\widetilde{L}_1 := \widetilde{T}L$ are nonnegative.
  Since the rows $r_1, \ldots, r_m$ of $\widetilde{L}_1$ belong to $C$, each of them can be represented as a nonnegative combination of the elements of $E$~\cite[\S 8.8]{convex}.
  For each $i = 1, \ldots, m$, we fix such a representation for $r_i$ and denote $E_i \subseteq E$ the set of elements of $E$ appearing in the representation with positive coefficients.
  We apply the generalized Hall's theorem~\cite[Theorem~1]{Hall} to the family $\mathcal{A} = \{E_1, \ldots, E_m\}$ of subsets of $E$ and the function $\mu$ such that $\mu(S)$ is defined to be the dimension of the linear span of the elements of $S$ for every $S \subseteq E$.
  This yields linearly independent elements $e_1, \ldots, e_m \in E$ such that $e_i \in E_i$ for every $i = 1,\ldots, m$.
  For every $i = 1, \ldots, m$, $r_i$ is a positive combination of $e_i$ and maybe some other elements of $E$, hence $\lVert e_i \rVert_0 \leqslant \lVert r_i \rVert_0$.
  Using~\eqref{eq:l0norm}, we have
  \[
    \lVert L_1 \rVert_0 \leqslant \sum\limits_{i = 1}^m \lVert e_i \rVert_0 \leqslant \sum\limits_{i = 1}^m \lVert r_i \rVert_0 = \lVert \widetilde{L}_1 \rVert_0,
  \]
  and this proves the minimality of the number of the nonzero entries in $L_1 = TL$ for $T$ constructed by the algorithm. 
\end{proof}

%
%

\section{Case Studies}
In this section, we demonstrate the improvements in physical intelligibility (while preserving the dimension) of reductions of biochemical models by our Algorithm~\ref{alg:main}. 
We analyse the results of the algorithm using models taken from the literature. 
We also compare the resulting reduction to the ones obtained by ERODE~\cite{erode} which are always defined by zero-one linear combinations.


    \subsection{Multisite protein phosphorylation}
    \label{subsec:case1}
    
	\subsubsection{Setup.} We consider a model of multisite phosphorylation \cite{Sneddon2010}.
	It describes a protein with $m$ identical and independent binding sites that simultaneously undergo phosphorylation and dephosphorylation.
	Each binding site can be in one of the four different states (see Figure~\ref{fig:multisite_setup}):
	\begin{enumerate}[topsep=1pt]
	    \item unphosphorylated and unbound, 
	    \item unphosphorylated and bound to a kinase,
	    \item phosphorylated and unbound,
	    \item phosphorylated and bound to a phosphatase.
	\end{enumerate}
    Therefore, there are $4^m + 2$ chemical species in the corresponding reaction network: free kinase and phosphatase, and $4^m$ states of the protein.
    
    \begin{figure}[H]
	\centering
    \includegraphics[width=0.8\textwidth]{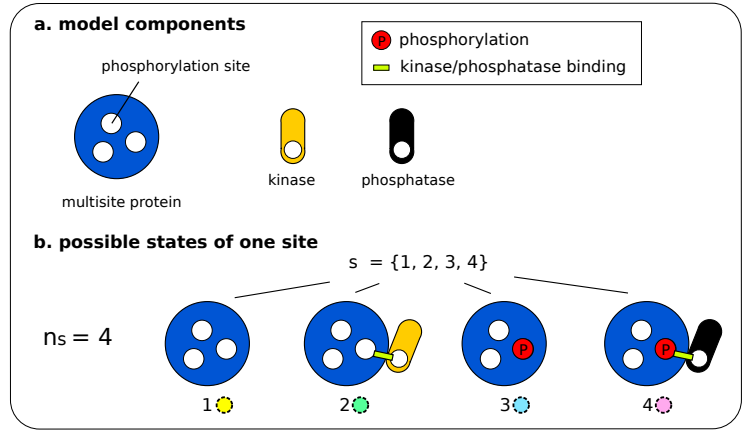}
    \caption{\small Molecular components and states of the multisite phosphorylation model. \\
    a. Consists of multisite proteins and kinases. This example has $m = 3$ sites.\\
    b. There are 4 possible states for a single site: unphosphorylated and unbound, unphosphorylated and bound to a kinase, phosphorylated and unbound, phosphorylated and bound to a phosphatase.}
    \label{fig:multisite_setup}
    \end{figure}
    
    \subsubsection{Reductions by ERODE and CLUE.} 
    In the reduction computed by ERODE~\cite{Cardelli2017} (for $m = 2, \ldots, 8$), the concentrations of protein configurations are replaced by the sums of the concentrations of configurations differing by a permutation of the sites.
    Therefore, the number of macro-variables is equal to $\binom{m + 3}{3} + 2$.
    
    In contrast, the analysis performed by CLUE~\cite{Ovchinnikov2021} always results in just six macro-variables.
    Two of them were always the concentrations of kinase and phosphatase as for ERODE.
    The other four were linear combinations with protein configurations.
    In~\cite[Section 4.2]{Ovchinnikov2021}, for $m = 2$, interpretation was provided for the first three of them. However, for the last one, it was remarked that ``the last macro-variable escaped physical intelligibility as it represents the difference between the free substrate with unphosphorylated sites and protein configurations that appear in the aforementioned lumps."
    
    \subsubsection{Our results.} We applied our algorithm to the cases $m = 2, 3, 4, 5$ and obtained new macro-variables, which have again included the concentrations of free kinase and phosphatase.
    Moreover, the three interpretable macro-variables from the analysis in~\cite{Ovchinnikov2021} for $m = 2$ are kept.
    Each of the four our macro-variables involving the protein configurations corresponds to a state of a site (e.g., unbounded and unphosphorylated), and each protein configuration appears with a coefficient equal to the number of sites in it with this state.
    Examples of these new macro-variables are given on Figure~\ref{fig:multisite_example} for $m = 2$ and $m = 3$.
    
    One way to interpret the result is that the constructed reduction replaces the concentration of the protein configurations with the ``concentrations'' of each of the four states of the sites (see also Example~\ref{ex:lumping}).
    From our interpretation, we expect that the models with larger $m$ will have a reduction of the same form.

    \begin{figure}[!htbp]
    \centering
    \includegraphics[width=0.9\textwidth]{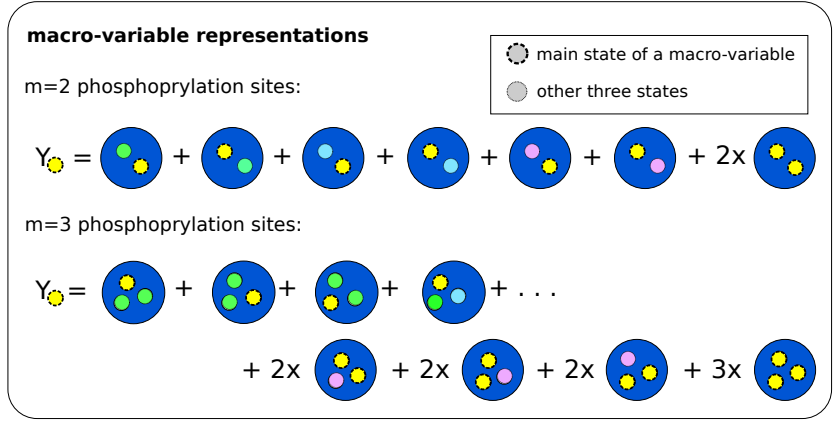}
    \caption{\small New macro-variables for $m = 2, 3$. 
    Each state of a binding site from Fig~\ref{fig:multisite_setup}-b can be the main state, yielding four macro-variables for each $m$. 
    The coefficients are equal to the number of binding sites in a protein that are in the main state.}
    \label{fig:multisite_example}
    \end{figure}


     \subsection{Fc$\epsilon$-RI signaling pathways}
    \label{subsec:fceri}
    
    \subsubsection{Setup.} 
    We consider a model for a different kind of multisite phosphorylation \cite{Faeder2003}, a model for the early events in the signaling pathway of the high-affinity IgE receptor (Fc$\epsilon$RI) in mast cells and basophils.
    
    The model details the rule-based interactions of Fc$\epsilon$RI receptor with a bivalent ligand (IgE dimer), the Src kinase Lyn, and the cytosolic protein tyrosine kinase Syk.
   The model is based on the following sequence of signaling events in Fc$\epsilon$RI \cite{Metzger2002,Nadler2001} (the reactions are nicely summarized on~\cite[Figure~2]{Faeder2003}):
    
    \begin{enumerate}[topsep=1pt]
        \item  binding of IgE ligand and Fc$\epsilon$RI which aggregates at the plasma membrane,
        \item transphosphorylation of tyrosine residues in the immunoreceptor tyrosine-based activation motifs (ITAMs) of the aggregated receptor by constitutively associated Lyn,
        \item recruitment of an extra Lyn/Syk kinase to the phosphorylated ITAM sites,
        \item transphosphorylation of Syk by Lyn and Syk on its linker region and activation loop, respectively.
    \end{enumerate}

    For visualizing different chemical species occurring in the resulting reaction network, we use the representation~\cite[Figure~1]{Faeder2003} summarized in Figure~\ref{fceri_setup}.  
    In total, there are $354$ of three types: monomers, dimers, and non-receptor states (free ligand/Lyn and Syk in each of 4 possible states of phosphorylation).
    
    \begin{figure}[H]
    \centering
    \includegraphics[width=0.7\textwidth]{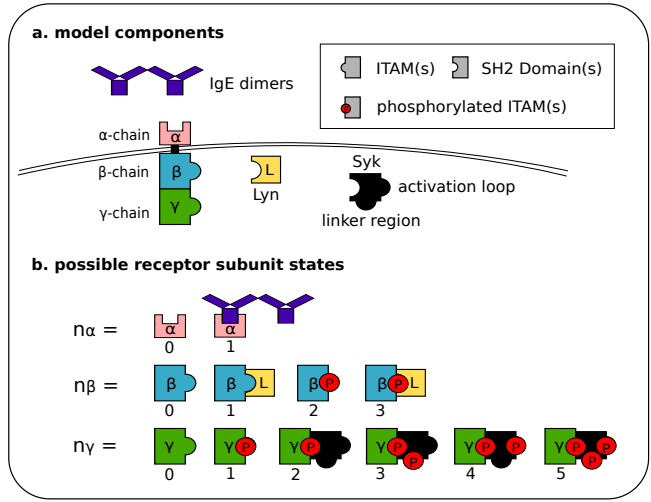}
    \caption{\small Molecular components and states of the Fc$\epsilon$RI signaling events model.\\
    a. IgE dimer is a bivalent ligand. Fc$\epsilon$RI consists of $\alpha$, $\beta$, $\gamma$ subunits. 
    Lyn kinase has an SH2 domain. Syk kinase has an SH2 domain and two ITAM sites which differ by the method of phosphorylation: Lyn at the linker region, and Syk at the activation loop.\\
    b. The $\alpha$ subunit can be unbound or bound to a ligand. 
    $\beta$ can be unphosphorylated/phosphorylated, with/without associated Lyn. 
    $\gamma$ can be unphosphorylated/phosphorylated, and the phosphorylated form binds to Syk in any of the four states of phosphorylation.}
    \label{fceri_setup}
    \end{figure}
    
    \begin{figure}[!htbp]
    \centering
    \includegraphics[width=0.8\textwidth]{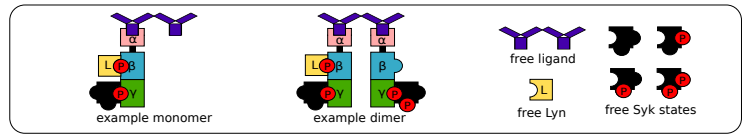}
    \caption{\small Examples of a monomer, a dimer, and the free components.}
    \label{fig7}
    \end{figure}	
    
    \vspace{-5mm}
    \subsubsection{Reductions by ERODE and CLUE.} 
    The reduction by ERODE~\cite{Cardelli2017} consists of $105$ macro-variables, where all the complexes with the same configuration except for the phosphorylation state of the Syk units are summed up in a single macro-variable.
    We will refer to these macro-variables as \emph{Syk-macro-variables}.
    
    The model has been reduced using CLUE in~\cite[Section~4.2]{Ovchinnikov2021} with the observable being the total concentration of the free Syk (in all the four phosphorylation states). 
    The reduced model had $69$ macro-variables, and $51$ of them were nontrivial.
    It has been observed in~\cite[Section~4.2]{Ovchinnikov2021} that some of these macro-variable carry a physical interpretation, but in some of them, negative elements were present, which hinder their physical intelligibility.
    
    \subsubsection{Our results.}
    We apply our algorithm to the reduced model.
    Among the new macro-variables, we have $51$ nontrivial macro-variables which is the same as for the CLUE reduction.
    More precisely, \ref{step:producing} produced $57$ nontrivial macro-variables, and this number has been reduced to $51$ when computing a linearly independent basis on~\ref{step:greedy}.
    The resulting macro-variable refine the reduction by ERODE mentioned above in the sense that our new macro-variable are the sums of the Syk-macro-variables with non-negative coefficients.
    Therefore, in our reduction, all the complexes differing only by the phosphorylation state of the Syk units are in the same macro-variable.
    For the \emph{monomers}, we obtain the same reduction: nontrivial macro-variable involving monomers are of the form described on Figure~\ref{fig:fceri_monomer}. 
    
    \begin{figure}[!ht]
    \includegraphics[width=\textwidth]{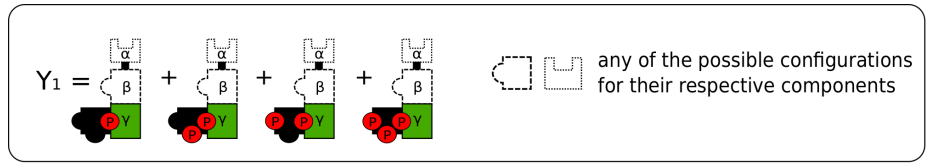}
    \centering
    \caption{The monomer macro-variables. In each of them, $\alpha$ and $\beta$ are fixed, and we sum over all the phosphorylation patterns of Syk.}
    \label{fig:fceri_monomer}
    \end{figure}	
    
    The macro-variable involving dimers are graphically described on Figure~\ref{fig:fceri_dimers}.
    First, one can see that they are indeed linear combinations of the Syk-macro-variables.
    Our interpretation of these new macro-variables is based on two observations about the set of the reactions in the original model~\cite[Figure~2]{Faeder2003}:
    \begin{description}
        \item[\textbf{(Obs. 1)}] For every reaction involving a dimer, only one of the receptors of the dimer is affected by the reaction.
        \item[\textbf{(Obs. 2)}] The $\gamma$-chain of the other (not affected) receptor is relevant only for the reactions of transphosphorylation of Syk.
    \end{description}
    Since the complexes with different phosphorylation patterns are grouped together in the Syk-macro-variables, the second observation implies that the transphosphorylation reactions do not affect the values of the Syk-macro-variables at all.
    Therefore, the first observation suggests considering macro-variables as sums over all the dimer configurations in which one receptor is fixed (up to the phosphorylation of Syk), and for the other receptor, all the possible variants of the $\gamma$-chain are considered.
    
    With this interpretation in mind, let us take a closer look at the Figure~\ref{fig:fceri_dimers}:
    \begin{itemize}
        \item Each of the variables as on Figure~\ref{fig:fceri_noSyk} is the sum over the configurations with the fixed left receptor not carrying Syk and the right receptor having each of the six possible $\gamma$-chains.
        If one of the complexes in the sum is fully symmetric, it appears with coefficient 2.
        \item Each of the variables as on Figure~\ref{fig:fceri_Syk} is a combination of complexes that have: the same $\beta$-chains and Syk on the left receptor, any phosphorylation pattern of the Syk on the left receptor, and any $\gamma$-chain on the right receptor.
        If the $\beta$-chains on the receptors are equal, the complexes with two Syk's (which are \emph{symmetric up to} Syk phosphorylation) appear with coefficient 2.
    \end{itemize}
    
    Note that the coefficients $2$ appearing in the presence of symmetry prevent ERODE~\cite{Cardelli2017} from finding this reduction.
    
    \begin{figure}[!htbp]
        
        \begin{subfigure}[b]{\textwidth}
        \centering
        \includegraphics[width=0.8\textwidth]{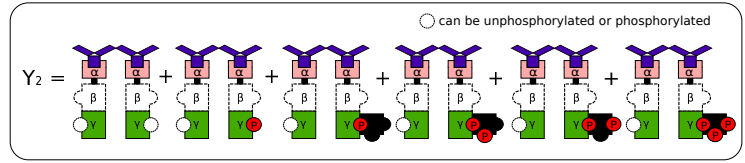}
         \caption{No Syk on the left $\gamma$-chain, the right takes all the possible $\gamma$-configurations}
         \label{fig:fceri_noSyk}
        \end{subfigure}    
        
        \begin{subfigure}[b]{\textwidth}
        \centering
        \includegraphics[width=0.8\textwidth]{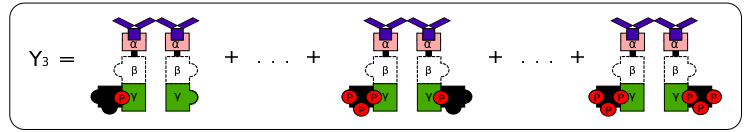}
    \caption{Syk (with all possible phosphorylations) on the left $\gamma$-chain, the right takes all the possible $\gamma$-configurations. }
    \label{fig:fceri_Syk}
        \end{subfigure}
        
        \begin{subfigure}[b]{\textwidth}
        \centering
        \includegraphics[width=0.8\textwidth]{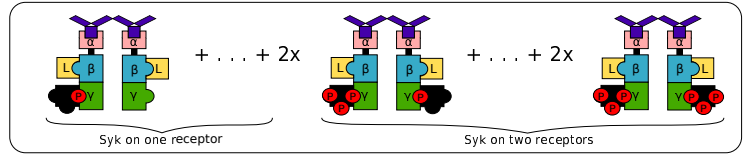}
    \caption{An example configuration of $Y_3$ (b) with identical $\beta$-chains. A variable has coefficient~2 when Syk kinases are recruited on both receptors (regardless of Syk state).}
    \label{fig:fceri_Syk2}
        \end{subfigure}
        \caption{Macro-variables involving dimers}
        \label{fig:fceri_dimers}
    \end{figure}


\subsection{Jak-family protein tyrosine kinase activation}
\label{subsec:barua}
\subsubsection{Setup.} 
We study a simplified cellular model of a bipolar ``clamp" mechanism for Jak-family kinase activation \cite{Barua2009}. 
Kinases of the Janus kinase (JAK) family play an essential role in signal transduction mediated by cell surface receptors, which lack innate enzymatic activities to dimerize. 
    
The model studies the interactions of Jak2 kinase trans-phosphorylation, specifically the rule-based dynamics between the Jak2 (J) kinase, the unique adaptor protein SH2-B$\beta$ (S) with the capacity to homo-dimerize, the growth hormone receptors (R), and a bivalent growth hormone ligand (L). The Jak2 kinase has two phosphorylation sites, Y1 and Y2. The SH2-B$\beta$ protein contains an N-terminal dimerization domain (DD) and a  C-terminal  Src homology-2  (SH2) domain. 

The components can interact in the following ways:
    \begin{enumerate}
        \item binding of ligand and growth hormone receptors which aggregates at the plasma membrane,
        \item constitutive binding of Jak2 kinase to the receptors, which autophosphorylates on the phosphorylation sites when two Jak2 kinases are bounded in the same complex,
        \item recruitment of SH2-B$\beta$ protein at the SH2 domain by the Jak2's autophosphorylated Y1 site,
        \item dimerization of SH2-B$\beta$ protein through recruitment of an additional SH2-B$\beta$ protein, engaged at the DD domains.
    \end{enumerate}

Receptors can undergo a process of \emph{internalization}, in which the receptors can no longer associate with any Jak2, and the existing Jak2 and SH2-B$\beta$ in the complex can dissociate at the normal rate. 
  \begin{figure}[!htbp]
     \centering
    \includegraphics[width=.8\textwidth]{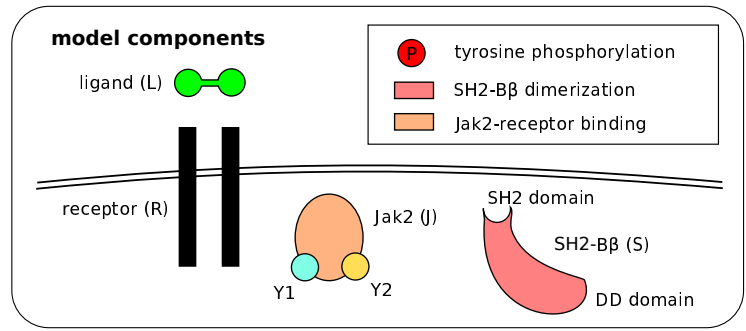}
    \caption{\small Components of the Jak-family protein tyrosine kinase activation model.\\}
    \label{barua_setup}
    \end{figure}

\subsubsection{Reductions by ERODE and CLUE.}
The reduction obtained by ERODE in~\cite[Figure~5]{Cardelli2017} contained $345$ macro-variables. 
It grouped the internalized configurations, which differ by the connections between the receptors and the ligand, into macro-variables.

The model has been reduced by CLUE in~\cite{Ovchinnikov2021}, with the observable being the concentration of the free ligand. 
The reduced model had $322$ macro-variables, and $69$ of them were nontrivial.
The model has been used in~\cite{Ovchinnikov2021} for benchmarking purposes only, so the macro-variables have not been interpreted.
The reduction included several macro-variables with negative coefficients, including one with 20 nonzero coefficients.
We do not see any natural interpretation for them.

\subsubsection{Our results.}  
We apply our algorithm to the reduced model, 
and among the produced macro-variables $69$ are nontrivial as in the reduction by CLUE (\ref{step:producing} produces $70$ macro-variables, and then this number is reduced to $69$ at~\ref{step:greedy}).
The nontrivial macro-variables are linear combinations of internalized molecules, and the trivial macro-variables are not internalized. 
Compared to the ERODE reduction, some internalized complexes such as the ligand-receptor (R, RL, RLR) structures in \cite[Figure~5(C)]{Cardelli2017}, are omitted in our model as they do not disassociate under internalization and thus do not affect the dynamics of the free ligand observable. 
In our reduction, mirrored internalized complexes are lumped together, which explains all two-element macro-variables.
The remaining nontrivial macro-variables are described on Figure~\ref{fig:barua}, and are of two types:
\begin{itemize}
    \item \emph{Configurations equivalent up to the connection between the ligand and the receptors (Figure~\ref{fig:barua}-a and~\ref{barua_fig2}).}
    The structures are equivalent under internalization as the ligand and receptors cannot disassociate and were obtained also by ERODE~\cite[Figure~5(D, E)]{Cardelli2017}. 
    They are of two types: single-Jak2-Receptor case (Figure~\ref{fig:barua}-a) and ``clamp'' case (Figure~\ref{barua_fig2}).

    \item \emph{Configurations with one receptor fixed (Figure~\ref{fig:barua}-b).}
    These are similar to Figure~\ref{fig:fceri_dimers} from the case study in Section~\ref{subsec:fceri}:
    since in the reactions with internalized complexes, only one receptor is affected, and this does not depend on the state of the other receptor, one can group together the complexes having one of the receptors the same.

    When the receptors are symmetric, the element and its mirrored element are the same, so the corresponding configuration appears with coefficient~2.
\end{itemize}

  \begin{figure}[!htbp]
    \centering
    \includegraphics[width=0.8\textwidth]{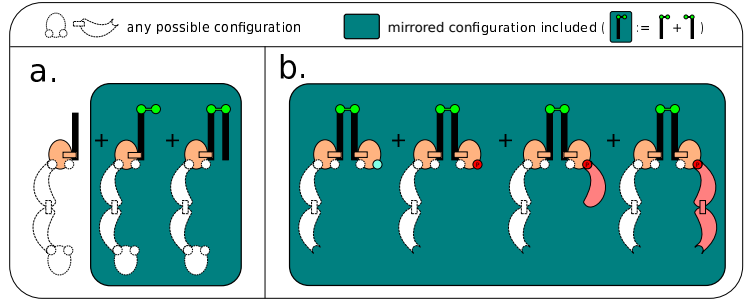}
    \caption{\small Classes of internalized macro-variables.
    a. equivalent up to the connection between the ligand and the receptors.
    b. one receptor fixed.}
    \label{fig:barua}
    \end{figure}
  \begin{figure}[!htbp]

    \includegraphics[width=0.9\textwidth]{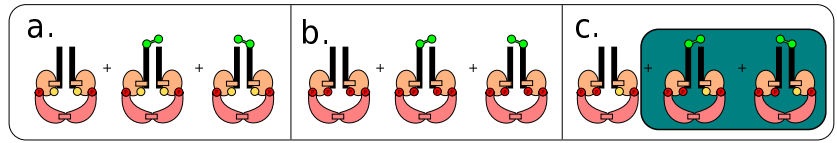}
    \caption{\small Macro-variables for equivalent structures of the bipolar "clamp" mechanism}
    \label{barua_fig2}
    \end{figure}


\section{Conclusion}

We have hypothesized that the interpretability of the macro-variables in an exact linear reduction may be improved by a change of coordinates making the macro-variable nonnegative combinations of the original variables and minimizing the number of nonzero coefficients.
We have designed and implemented an algorithm for performing such a transformation and applied it to three models (with hundreds of variables) for which the result of the reduction by CLUE~\cite{Ovchinnikov2021} contained macro-variables without a clear physical interpretation.
We have shown that the resulting macro-variables are interpretable, thus supporting the original hypothesis and demonstrating the usefulness of our algorithm.

Our results also give insight into the structure of reductions in which not all the coefficients are zeroes and ones.
In particular, we can point out two different situations:
\begin{itemize}
    \item The macro-variables are the ``concentrations'' of parts of molecules as in Section~\ref{subsec:case1}.
    The species having several identical pieces may appear with larger coefficients.
    \item Some of the molecules appearing in the macro-variable are symmetric (as in Section~\ref{subsec:barua}) or even partially symmetric (as in Section~\ref{subsec:fceri}), and they appear with a coefficient accounting for the symmetries.
\end{itemize}

\subsection*{Acknowledgement}

The authors are grateful to Fran\c{c}ois Fages, Mathieu Hemery, Sylvain Soliman, and Mirco Tribastone for helpful discussions and to the referees for helpful suggestions. 
GP is grateful to Heather Harrington, Gregory Henselman-Petrusek, and Zvi Rosen for educating him about the matroid theory.

%
%
%
\bibliographystyle{splncs04}
\bibliography{mybibliography}

\section*{Appendix: Non-positivizable reduction}
As we have mentioned in Remark~\ref{rem:NO}, we did not encounter examples from the literature for which Algorithm~\ref{alg:main} would return NO.
However, one can easily construct an artificial example with this property.
Consider the system
\begin{equation}
    \begin{cases}
      x_1' = x_1^2 + x_2^2,\\
      x_2' = 2x_1x_2.
    \end{cases}
\end{equation}
Then $y = x_1 - x_2$ yields a reduced system $y' = y^2$.
However, since any change of macro-variables is a scaling of $y$, there is no equivalent lumping with nonnegative coefficients, so Algorithm~\ref{alg:main} (with the input $L = (1\; -1)$) will return NO.

\end{document}